\documentclass[12pt]{article}
\usepackage{amsfonts}
\usepackage{amssymb}
\usepackage{amsmath}
\usepackage{cite,enumerate,amsthm,enumerate}

\usepackage{lscape}
\usepackage{diagbox}
\usepackage{multirow}
\usepackage{bigstrut}
\usepackage{float}

\newtheorem{theorem}{Theorem}[section]
\newtheorem{lemma}[theorem]{Lemma}
\newtheorem{proposition}[theorem]{Proposition}
\newtheorem{corollary}[theorem]{Corollary}

\newtheorem{example}[theorem]{Example}

\begin{document}

\markboth{Somphong Jitman and Supawadee Prugsapitak}
{Enumeration of Self-Dual Cyclic Codes of some Specific Lengths over Finite Fields}

 
%
%%%%%%%%%%%%%%%%%%%%%%%%%%%%%%%%%%%%%%%%%%%%%%%%%%%%%%%%%%%%%%%%%%%%

\title{Enumeration of Self-Dual Cyclic Codes of some Specific Lengths over Finite Fields}

\author{Supawadee  Prugsapitak\thanks{S. Prugsapitak is with the Department of Mathematics and Statistics, Faculty of Science,  Prince of Songkla University, Hat Yai,
    Songkhla 90110, Thailand
    {\tt email: supawadee.p@psu.ac.th}}
~ and Somphong Jitman\thanks{Corresponding Author}~$^,$\thanks{Department of Mathematics, Faculty of Science, Silpakorn University,  Nakhon Pathom 73000, Thailand
{\tt email: sjitman@gmail.com}}}

\maketitle

\begin{abstract}
Self-dual cyclic codes form an  important class of linear codes. It has been shown that there exists a self-dual cyclic code of length $n$ over a finite field if and only if $n$ and the field characteristic are even.   The enumeration of such codes has been given under both the Euclidean and Hermitian products. However, in each case, the formula for self-dual cyclic codes of length $n$ over a finite field contains a characteristic function which is not easily computed.  In this paper, we focus on more efficient ways to enumerate self-dual cyclic codes  of lengths $2^\nu p^r$ and  $2^\nu p^rq^s$,  where $\nu$, $r$, and $s$ are positive integers.   Some number theoretical tools are established.   Based on these results, alternative formulas and efficient algorithms to   determine  the number of self-dual cyclic codes of such  lengths   are provided.
\end{abstract}



    \section{Introduction}

Self-dual  cyclic codes  constitute an  important class of linear codes  due to their rich algebraic structures, their fascinating links to other objects, and their wide applications. Such codes have   been extensively studied for both theoretical and practical reasons (see \cite{CDL2014}, \cite{JLX2011}, \cite{JLS2013},   \cite{KZ2008},  \cite{NRS2006}, \cite{S2003} and references therein).
Some major results on Euclidean self-dual cyclic codes have been discussed in \cite{JLX2011} and \cite{KZ2008}. The complete characterization and enumeration of such codes have been established  in \cite{JLX2011}. These results have been generalized to the    Hermitian case  in \cite{JLS2013}.  Some related results on the enumeration of self-dual cyclic codes over finite chain rings can be found in \cite{CLZ2016}.

For a prime power $q$, denote by $\mathbb{F}_q$ the finite field of $q$ elements. A {\em linear code} $C$ of length $n$ over $\mathbb{F}_q$ is defined to be a subspace of the $\mathbb{F}_q$-vector space $\mathbb{F}_q^n$.  The Euclidean dual of a linear code $C$  is defined to be \[C^{\perp_{\rm E}}=\{\boldsymbol{v}\in \mathbb{F}_q^n\mid \langle \boldsymbol{v} , \boldsymbol{c} \rangle_{\rm E}=0\text{ for all } \boldsymbol{c}\in C\},\]
where $\langle \boldsymbol{v} , \boldsymbol{u} \rangle_{\rm E}:=\sum_{i=1}^n v_iu_i$ is the Euclidean inner product between $\boldsymbol{v} =(v_1,v_2,\dots, v_n) $ and  $\boldsymbol{u} =(u_1,u_2,\dots, u_n)$ in $\mathbb{F}_q^n$.
In the case where $q$ is a square, the Hermitian  dual of  a linear code $C$    can be defined as well and it  is defined to be \[C^{\perp_{\rm H}}=\{\boldsymbol{v}\in \mathbb{F}_{q}^n\mid \langle \boldsymbol{v} , \boldsymbol{c} \rangle_{\rm H}=0\text{ for all } \boldsymbol{c}\in C\},\]
where $\langle \boldsymbol{v} , \boldsymbol{u} \rangle_{\rm H}:=\sum_{i=1}^n v_iu_i^{\sqrt q}$ is the  Hermitian inner product between $\boldsymbol{v}   $ and  $\boldsymbol{u}  $ in $\mathbb{F}_{q}^n$.  A code $C$ is said to be {\em Euclidean self-dual}
(resp.  {\em Hermitian self-dual}) if  $C=C^{\perp _{\rm E}}$ (resp., $C=C^{\perp _{\rm H}}$).

A linear code $C$ is said to be {\em  cyclic } if it is invariant under the right cyclic shift.  It is well known that every cyclic code of length $n$ over $\mathbb{F}_q$ can be view as an (isomorphic) ideal in the principal ideal ring $\mathbb{F}_q[x]/\langle x^n-1\rangle$ uniquely generated by a monic divisor $g(x)$  of $x^n-1$. Such the polynomial is called the {\em generator polynomial} of $C$. For a polynomial $f(x)=\sum_{i=0}^kf_ix^i$  of degree $k$ in $\mathbb{F}_{q}[x]$ with $f_0\ne 0$,  the {\em reciprocal polynomial} of $f(x)$ is defined to be $f^*(x):=f_0^{-1} x^k \sum_{i=0}^kf_i(1/x)^i$.  If $q$ is a square, the {\em conjugate reciprocal polynomial} of $f(x)$ is defined to be $f^\dagger (x):=f_0^{-\sqrt{q}} x^k \sum_{i=0}^kf_i^{\sqrt{q}}(1/x)^i$.
In \cite{JLX2011} and \cite{JLS2013}, it has been shown that a cyclic code of length $n$  with the generator polynomial $g(x)$ is Euclidean self-dual (resp., Hermitian self-dual) if and only if $g(x)=h^*(x)$ (resp., $g(x)=h^\dagger(x)$), where $h(x)=\frac{x^n-1}{g(x)}$. Based on this results, the following characterizations for the existence of self-dual cyclic codes were  obtained in  \cite{JLX2011} and \cite{JLS2013}.
\begin{theorem}[{\cite{JLX2011}}]  There exists a Euclidean self-dual cyclic code of length $n$ over $\mathbb{F}_q$ if and only if $n$ and $q$ are even.
\end{theorem}

\begin{theorem} [{\cite{JLS2013}}] If $q$ is a square, then there exists an Hermitian  self-dual cyclic code of length $n$ over $\mathbb{F}_{q}$ if and only if $n$ and $q$ are even.
\end{theorem}

From the above characterizations, it suffices to study  Euclidean (resp., Hermitian) self-dual cyclic codes  of even length  $n=n^\prime 2^\nu$ over   $\mathbb{F}_{2^l} $ (resp., $\mathbb{F}_{2^{2l}} $), where $n^\prime$ is odd and $\nu$ is a positive integer.

The following functions are keys for determining the number of  self-dual cyclic codes. Let $\mathbb{O}$ denote the set of odd positive integers. For each positive integer $l$, let $\chi_l:\mathbb{O}\to \{0,1\}$ and $\lambda_l:\mathbb{O} \to \{0,1\}$ be functions defined by
\begin{align}
\label{chi}\chi_l (j) =
\begin{cases}
0 &\text{ if  there exists an integer } s\geq 1 \text{ such that }  j|(2^{ls}+1),\\
1 &\text{ otherwise,}
\end{cases}
\end{align}
and
\begin{align}
\label{lambda}\lambda_l (j)=
\begin{cases}
0&\text{ if  there exists an odd integer } s\geq 1 \text{ such that }  j|(2^{ls}+1),\\
1&\text{otherwise}.
\end{cases}
\end{align}

For coprime positive integers $i$ and $j$, let $\operatorname{ord}_{j}(i)$ denote the multiplicative order of $i$ modulo $j$.

The formulas for the number of Euclidean self-dual cyclic codes of length $n$ over $\mathbb{F}_{2^l}$ and the number of Hermitian self-dual cyclic codes of length $n$ over $\mathbb{F}_{2^{2l}}$  were given in \cite{JLX2011} and \cite{JLS2013} as follows.

\begin{theorem} [{\cite[Theorem 3]{JLX2011}}] \label{ESD-form}
    Let $l$ be a positive integer and let  $n=n^\prime 2^\nu$ be a positive integer such that $n^\prime \geq 1$ is odd and $\nu \geq 1$. Then the number of Euclidean  self-dual cyclic codes of length $n$ over $\mathbb{F}_{2^{l}}$ is
    \begin{align}
    (2^\nu+1)^{\frac{1}{2}\sum_{d| n^\prime}\chi_l(d)\frac{\phi(d)}{ {\rm ord}_{d}(2^{l})}}.\label{ESDcc}
    \end{align}
\end{theorem}

\begin{theorem} [{\cite[Corollary 3.7]{JLS2013}}]  \label{HSD-form}
    Let $l$ be a positive integer and let $n=n^\prime 2^\nu$ be a positive integer such that $n^\prime \geq 1$ is odd and $\nu \geq 1$. Then the number of Hermitian self-dual cyclic codes of length $n$ over $\mathbb{F}_{2^{2l}}$ is
    \begin{align}
    (2^\nu+1)^{\frac{1}{2}\sum_{d| n^\prime}\lambda_l(d)\frac{\phi(d)}{ {\rm ord}_{d}(2^{2l})}}.\label{HSDcc}
    \end{align}
\end{theorem}

From the above theorems, the difficulty is to compute
\begin{align}\label{eq:t} t(n',l) := \frac{1}{2} \sum_{d | n'} \chi_l(d)\frac{\phi(d)}{\operatorname{ord}_{d}(2^l)}\end{align}
and
\begin{align}\label{eq:tau} \tau (n',l) := \frac{1}{2} \sum_{d | n'} \lambda_l(d)\frac{\phi(d)}{\operatorname{ord}_{d}(2^{2l})}.\end{align}
We can see that  the formulas contain the functions $\operatorname{ord}_{d}$, $\chi_l$ and $\lambda_l$ which are possible but not easy to determined. In this paper, we focus on Euclidean and Hermitian self-dual cyclic codes of  some specific lengths and aim to reduce the complexity in  computing $ t(n',l) $ and $\tau (n',l) $ in \eqref{eq:t} and \eqref{eq:tau}, respectively.  Precisely, we focus on the enumeration of self-dual cyclic codes of length $2^\nu p^r$ and  $2^\nu p^rq^s$ with respect to the Euclidean and Hermitian inner products, where $p$ and $q$ are distinct odd primes and $\nu$, $r$, and $s$ are positive integers.

After this introduction,  some number theoretical tools and an efficient algorithm to determine the  number of Euclidean  self-dual cyclic codes of length $2^\nu p^r$ and  $2^\nu p^rq^s$  over $\mathbb{F}_{2^l}$ are given in Section 2.    The analogous results for the   Hermitian case are given  in Section 3.

%%%%%%%%%%%%%%%%%%%%%%%%%%%%%%%%%%%%%%%%%%%%%%%%%%%%%%%%%%%%%
%%%%%%%%%%5   Euclidean Self-Dual Cyclic Codes    %%%%%%%%%%%%%%%%%%%%%%%%%%%%%%%%%%%
%%%%%%%%%%%%%%%%%%%%%%%%%%%%%%%%%%%%%%%%%%%%%%%%%%%%%%%%%%%%%

\section{Euclidean Self-Dual Cyclic Codes }

In this section,  number theoretical tools and efficient algorithms  for determining  the formula for Euclidean self-dual cyclic codes of length $2^\nu p^r$ and $2^\nu p^r q^s$ over $\mathbb{F}_{2^l}$ in Theorem  \ref{ESD-form}  are  given.  From  Theorem  \ref{ESD-form}, it is sufficient to focus on the value of  $t (n',l)$ in \eqref{eq:t}, where $ n^\prime\in \{p^r,p^rq^s\}$.

\subsection{Number Theoretical Results}
In order to give an efficient way to compute the number of  Euclidean self-dual  cyclic codes,  we begin with  the following number theoretical results.

For a prime $p$ and  integers $i\geq 0$ and  $j\geq 1$, we say that $p^i$ exactly divides $j$, denoted by $p^i || j$, if $p^i$ divides $j$ but $p^{i+1}$ does not divide $j$.

\begin{lemma}\label{lem2.1} Let $p$ be an odd prime and let $l$ be a positive integer. Let $\gamma$ and $i$ be  the  integers such that $2^\gamma || \operatorname{ord}_{p}(2) $ and $2^i||l$.  Then one of the following statements holds.
    \begin{enumerate}
        \item   $i<\gamma$ if and only if  $2^{\gamma-i}|| \operatorname{ord}_{p}(2^l)$.
        \item  $i\geq \gamma$ if and only if $\operatorname{ord}_{p}(2^l)$ is odd.
    \end{enumerate}
\end{lemma}
\begin{proof}
    We first recall that
    \begin{align}\label{order}\operatorname{ord}_{p}(2) = \operatorname{gcd}(\operatorname{ord}_{p}(2),l) \operatorname{ord}_p(2^l).\end{align}
    Let $j$ be a nonnegative integer such that $2^{j} ||  \operatorname{ord}_p(2^l)$.
    By considering the highest power of two that appear in (\ref{order}), we have $j = \gamma -\operatorname{min}(i,\gamma)$. This completes the proof.
\end{proof}
\begin{corollary}\label{lem:odd-gamma}
    Let $p$  be an odd prime and let $l$ be a positive integer.   Let     $\gamma$ be  the    integer such that $2^\gamma || \operatorname{ord}_{p}(2)$.  Then the following statements holds.
    \begin{enumerate}
        \item  $\operatorname{ord}_{p}(2^l)$ is odd if and only if    $2^\gamma | l$.
        \item If $\gamma\geq 1$, then $2||\operatorname{ord}_{p}(2^l)$ if and only if      $2^{\gamma-1}||l$.
        \item   If $\gamma\geq 1$, then  $4|\operatorname{ord}_{p}(2^l)$ if and only if    $2^{\gamma-1}\nmid l$.
    \end{enumerate}
\end{corollary}
\begin{proof}  Let $i$ be an integer such that $2^i||l$. Then the first part is immediately deduced from Lemma \ref{lem2.1} since $2^i ||l$.

    If $2^{\gamma-1} || l$, then  $i = \gamma-1 < \gamma.$ Again by Lemma \ref{lem2.1}, $2 || \operatorname{ord}_p(2^l)$.
    Conversely, if $2 || \operatorname{ord}_p(2^l)$ which implies that $\operatorname{ord}_p(2^l)$ is even, then by Lemma \ref{lem2.1} we deduced
    that $i < \gamma$ and hence $\gamma-i = 1$. Thus $2^{\gamma-1} || l.$ The second part is proved.
    
    If $2^{\gamma-1} \nmid l $, then $\gamma-1 > i$ which means $\gamma - i \geq 2$ and $\gamma > i$. Thus by  Lemma \ref{lem2.1}, $2^2 | \operatorname{ord}_p(2^l)$.
    Conversely, if $4 | \operatorname{ord}_p{2^l}$, then $\gamma-i \geq 2$ and this implies that  $2^{\gamma-1} \nmid l$. This completes the third part.
\end{proof}

%    In 1958. W. Sierpinski \cite{WS1958} proved that prime divisors of integers of the form $2^r+1$ for positive integer $r$ are all primes of the form $8n\pm 3$ and infinitely many primes of the form
%    $8n+1$ and prime divisors of integers of the form $2^{2r}-1$ are all primes of the form $8n+ 1$. Later in 1959, A. Brauer \cite[Theorem 1]{AB1959} proved that $p$ is  a prime divisor of integers of the form $a^r+1$
%    (resp., $a^{2r+1}-1$) if and only if $\operatorname{ord}_p(2)$ is even (resp., odd). Therefore we can conclude that if $p \equiv \pm3 \mod{8}$, then $\operatorname{ord}_p(2)$ is even ,  and if $p \equiv  7 \mod{8}$, then $\operatorname{ord}_p(2)$ is odd. However,  $\operatorname{ord}_p(2)$ is either odd or even if $p \equiv 1 \mod{8}$.  Based on these properties, Lemma \ref{lem2.1},  an Euler's criterion $2^{\frac{p-1}{2}} = (-1)^{\frac{p^2-1}{8}}\mod{p}$, we have the following corollary  (c.f. \cite[Proposition 2]{JLX2011}).
%
%    \begin{corollary}
%        Let $p$  be an odd prime and let $l$ be a positive integer. Then the following statements hold.
%        \begin{enumerate}
%            \item If  $p\equiv 3 \mod{8}$, then $\operatorname{ord}_{p}(2^l)$ is odd if and only if  $ 2|l$.
%            \item If  $p\equiv 5 \mod{8}$, then $\operatorname{ord}_{p}(2^l)$ is odd if and only if  $4|l$.
%            \item If  $p\equiv 7 \mod{8}$, then $\operatorname{ord}_{p}(2^l)$ is odd.
%        \end{enumerate}
%    \end{corollary}

For an odd integer $d>1$, necessary and sufficient conditions for  $\chi_l(d)$ to be zero  were  determined in  \cite{M1997}  in terms of $\operatorname{ord}_{p}(2^l)$, where $p$ is a prime divisor of $d$.
\begin{lemma}
    [{\cite[Theorem 1]{M1997}}]\label{goodP} Let $d>1$ be an odd integer and let $l$ be a positive integer. Then $\chi_l(d)=0$   if and only if there exists $e\geq 1$ such that $2^e|| {\rm ord}_{p}(2^l)$ for every prime $p$ dividing $d$.
\end{lemma}

{Generally, for any positive integer  $l$ and any odd integer $d>1$, the value of $\chi_l(d)$ can be obtained by considering a parity of
    $\operatorname{ord}_{p}(2^l)$ for each prime divisor  $p$ of $d$. It is easy to see that the following corollary holds.}

\begin{corollary}\label{char:good} Let $p$  be an odd prime and let $l$ be a positive integer. Then the following statements hold.
    \begin{enumerate}
        \item $\chi_l(p)=1$ if and only if $\operatorname{ord}_{p}(2^l)$ is odd.
        \item $\chi_l(p)=0$ if and only if $\operatorname{ord}_{p}(2^l)$ is even.
        \item $\chi_l(p^i)=\chi_l(p)$   for all positive integers $i$.
        % \item $\chi_l(pq)=1$ if and only if
    \end{enumerate}
\end{corollary}

%%%%%%%%%%%%%%%%%%%%%%%%%%%%%%%%%%%%%%%%%%%%%%%%%%%%%%%%%%%%%%%%%%%%%%%%%
%%%%%%%%%%%%%%%%%%%%%%         Lemma 2.5
%%%%%%%%%%%%%%%%%%%%%%%%%%%%%%%%%%%%%%%%%%%%%%%%%%%%%%%%%%%%%%%%%%%%%%%%%
The values $\operatorname{ord}_{p^i}(2^l)$ for all $1\leq i\leq r$ play a vital  role in determining $t(p^r,l)$. Here, we  simplify $\operatorname{ord}_{p^i}(2^l)$   in terms of $\operatorname{ord}_{p}(2^l)$.
\begin{lemma}\label{maxalpha}
    Let $p$ be an odd prime  and let $l$ and $i$ be  positive integers.  If  $\alpha$  is  the largest  integer such that $p\nmid  \operatorname{ord}_{p^\alpha}(2^l)$, then \[ \operatorname{ord}_{p^i}(2^l)=
    \begin{cases}
    \operatorname{ord}_{p}(2^l) &\text{ if } i\leq \alpha,   \\
    p^{i-\alpha}\operatorname{ord}_{p}(2^l)    &\text{ if } \alpha< i.
    \end{cases}
    \]
    In particular, if $2$ is a primitive root modulo $p^2$, then
    \[ \operatorname{ord}_{p^i}(2^l)=    \frac{p^{i-1}(p-1) }{\gcd(p^{i-1}(p-1),l)}
    \]
    for all positive integers $i$.
\end{lemma}

\begin{proof}
    {
        The first part of this lemma follows from   \cite[Theorem 3.6]{N2000}.
        
        Next, we assume that $2$ is a primitive root modulo $p^2$. Then $2$
        is a primitive root modulo $p$ and $2$
        is a primitive root modulo $p^i$ for $j \geq 2$ by the  Primitive Element Theorem.  In other words,
        $ \operatorname{ord}_{p^i}(2^l)=    \frac{p^{i-1}(p-1) }{\gcd(p^{i-1}(p-1),l)}    $
        for all positive integers~$i$.
    }
\end{proof}

\begin{corollary} \label{lem:alpha}
    Let $p$ be an odd prime  and let $l$ and $i$ be  positive integers.  If  $p^i|| \operatorname{ord}_{p^r}(2^l)$,  then  $\alpha=r-i$  is the largest
    integer such that $p \nmid  \operatorname{ord}_{p^\alpha}(2^l)$.
\end{corollary}

{
    \begin{proof}
        By Lemma \ref{maxalpha}, it is easy to see that for any $1 \leq j \leq \alpha$,
        $\operatorname{ord}_{p^i}(2^l)$ is not divisible by $p$ and for any $i \geq 1$, $p^i || \operatorname{ord}_{p^{\alpha+i}}(2^l)$.
        
        By the assumption, we have $p^i || \operatorname{ord}_{p^{r}}(2^l)$. Thus $r=\alpha + i$ and hence
        the largest integer such that $p\nmid  \operatorname{ord}_{p^\alpha}(2^l)$ is $r-i$ as desired.
\end{proof}}

From  Corollary \ref{lem:alpha},   the  integer $\alpha$ can be computed. Hence,  for each $1\leq i\leq r$,  $\operatorname{ord}_{p^i}(2^l)$  follows from
Lemma \ref{maxalpha}.

\subsection{Euclidean Self-Dual Cyclic Codes of Length $2^\nu p^r$}

In this subsection,  an alternative and simplified formula for   Euclidean  self-dual cyclic codes of length $2^\nu p^r $  over $\mathbb{F}_{2^l}$ is given based on the number theoretical tools given  in Subsection 2.1. An efficient  algorithm to compute the number of such self-dual codes is provided as well.

\begin{theorem}\label{tvalue} Let $p$ be an odd prime and let $l$ and $r$ be positive integers. Let $\alpha $ be the largest  integer such that $p\nmid \operatorname{ord}_{p^\alpha}(2^l)$ and let $\gamma $ be the    integer such that $2^\gamma  ||  \operatorname{ord}_{p}(2)$ . Then
    %        \begin{align}\label{eq:simt-old}
    %        t(p^r,l) =\begin{cases}
    %        \frac{ \gcd(\operatorname{ord}_{p}(2),l)}{2{\operatorname{ord}_{p}(2)}} \left(  {p^{\alpha}-1}+  \frac{p^\alpha(p^{r-\alpha}-1)}{p^{r-\alpha }}  \right) & \text{ if }   2^\gamma |l,\\
    %        0& \text{ if }    2^\gamma \nmid l.
    %          \end{cases}\end{align}
    
    {    \begin{align}\label{eq:simt}
        t(p^r,l) =\begin{cases}
        \frac{ \gcd(\operatorname{ord}_{p}(2),l)}{2{\operatorname{ord}_{p}(2)}} \left(  {p^{\alpha}-1} +(p-1)(r-\alpha) p^{\alpha-1}    \right) & \text{ if }   2^\gamma |l,\\
        0& \text{ if }    2^\gamma \nmid l.
        \end{cases}\end{align}}
    In particular, if $2$ is a primitive root modulo $p^2$, then for any positive integer $r$
    \begin{align} \label{eq:simtPri}
    t(p^r,l) = \begin{cases}
    \frac{1}{2}\left( \sum_{i=1}^{r} \operatorname{gcd}(p^{i-1}(p-1),l)\right)   & \text{ if }    2^\gamma |l,\\
    0& \text{ if }    2^\gamma\nmid l.
    \end{cases}\end{align}
\end{theorem}

\begin{proof}   Assume that $ 2^\gamma \nmid l$.  Then $\operatorname{ord}_{p}(2^l)$ is even by Corollary \ref{lem:odd-gamma}.  By  Corollary \ref{char:good}, we have  $\chi_l(p^i)=0$ for all $1\leq i\leq r$.   It follows that  $t(p^r,l) =0$.

    Next, assume that    $  2^\gamma |l$.  By Corollary \ref{lem:odd-gamma},  it follows that $\operatorname{ord}_{p}(2^l)$ is odd.   Hence,  by  Corollary \ref{char:good},  $\chi_l(p^i)=1$ for all $1\leq i\leq r$.
    From \eqref{eq:t} and Lemma \ref{maxalpha}, it can be concluded that
    \begin{align} t(p^r,l) &= \frac{1}{2} \sum_{i=0}^{r} \chi_l(p^i)\frac{\phi(p^i)}{\operatorname{ord}_{p^i}(2^l)}\\
    &=\frac{1}{2} \sum_{i=1}^{r} \frac{p^{i-1}(p-1)}{\operatorname{ord}_{p^i}(2^l)}\\
    &=\frac{1}{2} \left(\sum_{i=1}^{\alpha} \frac{p^{i-1}(p-1))}{\operatorname{ord}_{p^i}(2^l)} +  \sum_{i=\alpha+1}^{r} \frac{p^{i-1}(p-1)}{\operatorname{ord}_{p^i}(2^l)}\right)\\
    %        &=\frac{1}{2} \left(\sum_{i=1}^{\alpha} \frac{p^{i-1}(p-1)}{\operatorname{ord}_{p}(2^l)} +  \sum_{i=\alpha+1}^{r} \frac{p^{i-1}(p-1)}{p^{r-\alpha }\operatorname{ord}_{p}(2^l)}\right)\\
    &={ \frac{1}{2} \left(\sum_{i=1}^{\alpha} \frac{p^{i-1}(p-1)}{\operatorname{ord}_{p}(2^l)} +  \sum_{i=\alpha+1}^{r} \frac{p^{i-1}(p-1)}{p^{i-\alpha }\operatorname{ord}_{p}(2^l)}\right)}\\
    %
    %        &=  \frac{p-1}{2{\operatorname{ord}_{p}(2^l)}} \left(\sum_{i=1}^{\alpha} {p^{i-1}}+  \sum_{i=\alpha+1}^{r} \frac{p^{i-1}}{p^{r-\alpha } }\right)\\
    &=  { \frac{p-1}{2{\operatorname{ord}_{p}(2^l)}} \left(\sum_{i=1}^{\alpha} {p^{i-1}}+  \sum_{i=\alpha+1}^{r} {p^{\alpha -1}}\right)}\\
    %
    %       &=  \frac{(p-1) \gcd(\operatorname{ord}_{p}(2),l)}{2{\operatorname{ord}_{p}(2)}} \left(\sum_{i=1}^{\alpha} {p^{i-1}}+  \frac{p^\alpha}{p^{r-\alpha }}\sum_{i=1}^{r-\alpha } {p^{i-1}} \right)\\
    &= { \frac{(p-1) \gcd(\operatorname{ord}_{p}(2),l)}{2{\operatorname{ord}_{p}(2)}} \left(\sum_{i=1}^{\alpha} {p^{i-1}}+  (r-\alpha)p^{\alpha-1}  \right)}\\
    %
    %        &= \frac{(p-1) \gcd(\operatorname{ord}_{p}(2),l)}{2{\operatorname{ord}_{p}(2)}} \left( \frac {p^{\alpha}-1}{p-1}+  \frac{p^\alpha(p^{r-\alpha}-1)}{p^{r-\alpha }(p-1)}  \right)\\
    &= {\frac{(p-1) \gcd(\operatorname{ord}_{p}(2),l)}{2{\operatorname{ord}_{p}(2)}} \left( \frac {p^{\alpha}-1}{p-1} +(r-\alpha)p^{\alpha-1}    \right)}\\
    %
    %        &= \frac{ \gcd(\operatorname{ord}_{p}(2),l)}{2{\operatorname{ord}_{p}(2)}} \left(  {p^{\alpha}-1}+  \frac{p^\alpha(p^{r-\alpha}-1)}{p^{r-\alpha }}  \right)\\
    &= {\frac{ \gcd(\operatorname{ord}_{p}(2),l)}{2{\operatorname{ord}_{p}(2)}} \left(  {p^{\alpha}-1}+(p-1)(r-\alpha)p^{\alpha-1} \right).}
    \end{align}
    
    Finally, assume that $2$ is a primitive root modulo $p^2$. By Lemma \ref{maxalpha}, we have   $\operatorname{ord}_{p^i}(2^l)=    \frac{p^{i-1}(p-1) }{\gcd(p^{i-1}(p-1),l)}$ for all integers $1\leq i \leq r$.   Hence,
    \begin{align} t(p^r,l) &= \frac{1}{2} \sum_{i=0}^{r} \chi_l(p^i)\frac{\phi(p^i)}{\operatorname{ord}_{p^i}(2^l)}\\
    &=\frac{1}{2} \sum_{i=1}^{r} \frac{p^{i-1}(p-1)}{ \frac{p^{i-1}(p-1) }{\gcd(p^{i-1}(p-1),l)}}\\
    &= \frac{1}{2} \sum_{i=1}^{r} \gcd(p^{i-1}(p-1),l)
    \end{align}
    as desired.
\end{proof}
From the above theorem, we obtain the following corollary.
\begin{corollary}
    If $p$ is an odd prime and $l, r$ are positive integers satisfying $l$ is not divisible by $p$, $\operatorname{ord}_p{(2^l)}$ is odd and $2$ is a primitive root modulo $p^2$, then
    $t(p^r,l) = \frac{r}{2}\operatorname{gcd}(p-1,l)$.
\end{corollary}

\begin{example}\label{example_t} Let $p=11$, $l=4$ and $r=3$. We can see that $2$ is a primitive root modulo $11^2$ and $\operatorname{ord}_{11}{(2^4)} = 5$ is odd.
    Thus $t(11^3,4) = \frac{3}{2}\operatorname{gcd}(10,4)$=3.
\end{example}

The results discussed above can be  summarized in  Algorithm \ref{al1}.

\begin{figure}[!hbt]
    \centering
    \parbox{12cm}{ \hrule \vskip1em

        For  an odd prime $p$ and positive integers $l$ and $r$, do the following steps.
        \begin{enumerate}
            \item     Compute $\operatorname{ord}_{p}{(2)}$.
            \item   Determine $\gamma$ such that  $2^\gamma || \operatorname{ord}_{p}{(2)}$.
            \begin{enumerate}
                \item[2.1]  If $2^\gamma\nmid l$, then $t(p^r,l)=0$ by \eqref{eq:simt}. Done.
                \item[2.2] If $2^\gamma| l$, then compute $\operatorname{ord}_{p^2}{(2)}$.
                \begin{enumerate}
                    \item[2.2.1]  If  $\operatorname{ord}_{p^2}{(2)}=p(p-1)$, then evaluate \eqref{eq:simtPri}. Done.
                    \item[2.2.2]   If  $\operatorname{ord}_{p^2}{(2)}\ne p(p-1)$, then do the following steps.
                    \begin{enumerate}[i)]
                        \item Compute $\operatorname{ord}_{p^r}{(2^l)}$.
                        \item Determine  the largest
                        integer $\alpha$ such that $p \nmid  \operatorname{ord}_{p^\alpha}(2^l)$  (by Corollary \ref{lem:alpha}).
                        \item Evaluate   \eqref{eq:simt}. Done.
                    \end{enumerate}
                \end{enumerate}
            \end{enumerate}
        \end{enumerate}
        \hrule
    }
    \caption{Steps in  Computing $t(p^r,l)$}     \label{al1}
\end{figure}

To compute $t(p^r,l)$ directly  from \eqref{eq:t}, we need to compute $\chi_l(p^i)$, $\phi(p^i)$, and $\operatorname{ord}_{p^i}{(2^l)}$ for all $0\leq i\leq r$.  It is not difficult to see that  Algorithm \ref{al1} can reduce some complexity since it  requires to compute only $\operatorname{ord}_{p}{(2)}$, $\operatorname{ord}_{p^2}{(2)}$, $\operatorname{ord}_{p^r}{(2^l)}$ and some basic expressions in \eqref{eq:simt} or \eqref{eq:simtPri}.

\subsection{Euclidean Self-Dual Cyclic Codes of Length $2^\nu p^rq^s$}

In this subsection,  an alternative and simplified formula for   Euclidean  self-dual cyclic codes of length $2^\nu p^rq^s $  over $\mathbb{F}_{2^l}$ is given  as well as an efficient  algorithm to compute the number of such self-dual codes.

%{\color{blue} We next discuss on the value of $\chi_l(j)$ for a positive integer $l$ and an odd positive integer $j$.
%A necessary and sufficient condition is given in \cite{JLX2011}, namely $\chi_l(j) = 0$ if and
%only if there exists $e \geq 1$ such that $2^e || \operatorname{ord}_p{(2^l)}$ for every prime $p$ dividing $j$. }
%Let $\mu(n)$ be the highest power of $2$ dividing $n$. The following lemma follows immediately from the mentioned result.
%%%%%%%%%%%%%%%%%%%%%%%%%%%%%%%%%%%%%%%%%%%%%%%%%%%%%%%%%%%%%%%%
%%%  CHI P Q
%%%%%%%%%%%%%%%%%%%%%%%%%%%%%%%%%%%%%%%%%%%%%%%%%%%%%%%%%%%%%%%%%

The following lemma is a key to simplified the formula of $t(p^rq^s,l)$.

\begin{lemma}\label{chipq}
    Let $p$ and $q$  be  distinct odd primes and let $l$ be a positive integer. Let $\gamma$ and $\beta$ be the  integers such that  $2^\gamma|| \operatorname{ord}_p{(2)}$ and  $2^\beta|| \operatorname{ord}_q{(2)}$, respectively. Then the following statements hold.
    \begin{enumerate}
        \item $\chi_l(pq)=1$ if and only if one of the following statements holds.
        \begin{enumerate}
            \item $\chi_l(p)=1$ or $\chi_l(q)=1$.
            \item  $\chi_l(p)=0 = \chi_l(q)$  and  $\gamma\ne \beta$.
        \end{enumerate}
        \item $\chi_l(pq)=1$ if and only if $\chi_l(p)=0 = \chi_l(q)$  and  $\gamma= \beta$.
        \item $\chi_l(p^iq^j)=\chi(pq)$ for all positive integers $i$ and $j$.
    \end{enumerate}
\end{lemma}
\begin{proof}
    To prove the first part, let  $\gamma^\prime $ and $\beta^\prime$ be the  integers such that  $2^{\gamma^\prime} || \operatorname{ord}_p{(2^l)}$ and  $2^{\beta^\prime}|| \operatorname{ord}_q{(2^l)}$, respectively. Assume that  $\chi_l(pq)=1$.    By Lemma \ref{goodP},   it follows that  1) either $\operatorname{ord}_p{(2^l)}$ or  $\operatorname{ord}_p{(2^l)}$ is odd, or 2)  $\operatorname{ord}_p{(2^l)}$  and   $\operatorname{ord}_p{(2^l)}$ are even and  $\gamma^\prime \ne \beta^\prime$  by  Corollary \ref{lem:odd-gamma}. The former implies that $\chi_l(p)=1$ or $\chi_l(q)=1$. The latter implies that  $\chi_l(p)=0 = \chi_l(q)$ and  $\gamma^\prime \ne \beta^\prime$.   Since  $\operatorname{ord}_p{(2^l)}$  and   $\operatorname{ord}_p{(2^l)}$ are even,  Lemma \ref{lem2.1} implies that $\gamma = \gamma'-i$ and $\beta = \beta'-i$. Thus $\gamma\ne \beta$.
    %%
    %%    Conversely, assume that    the statement (a) or (b) holds.
    %%If $\chi_l(p)=1$ or $\chi_l(q)=1$, then  $\chi_l(pq)=1$ by  Lemma \ref{goodP}. Assume that  $\chi_l(p)=0 = \chi_l(q)$  and  $\gamma\ne \beta$.  Since  $\chi_l(p)=0 = \chi_l(q)$,  we have $\gamma^\prime >0$ and $\beta^\prime >0$ by Corollary \ref{char:good}.     From Lemma \ref{lem2.1},  we have
    %%\begin{align} \label{eq:gam-be2} {2^{\gamma^\prime}}{\gcd(2^\gamma, l)}=2^{\gamma}\ne 2^{\beta} ={2^{\beta^\prime}}{\gcd(2^\beta, l)}.
    %%\end{align}
    %%Since  $\operatorname{ord}_p{(2^l)}$  and   $\operatorname{ord}_p{(2^l)}$ are even,    $2^\gamma \nmid l$  and   $2^\beta \nmid l$ by  Corollary \ref{lem:odd-gamma}. It follows that $\gcd(2^\gamma, l)=\gcd(2^\beta, l)$. Hence, from  Equation \eqref{eq:gam-be2}, we have $\gamma^\prime\ne \beta^\prime$. Therefore,  $\chi_l(pq)=1$ as desired.

    Conversely, assume that    the statement (a) or (b) holds.
    If $\chi_l(p)=1$ or $\chi_l(q)=1$, then  $\chi_l(pq)=1$ by  Lemma \ref{goodP}. Assume that  $\chi_l(p)=0 = \chi_l(q)$  and  $\gamma\ne \beta$.  Since  $\chi_l(p)=0 = \chi_l(q)$,  we have $\gamma^\prime >0$ and $\beta^\prime >0$ by Corollary \ref{char:good}. Since  $\operatorname{ord}_p{(2^l)}$  and   $\operatorname{ord}_p{(2^l)}$ are even,  Lemma  \ref{lem2.1} implies that $\gamma' =\gamma-i$ and $\beta' =\beta-i$. Thus  $\gamma^\prime\ne \beta^\prime$. Therefore,  $\chi_l(pq)=1$ as desired.
    
    It is not difficult to see that  the second part and the first part are equivalent and the third one follows from Lemma \ref{goodP}.
\end{proof}

%%%%%%%%%%%%%%%%%%%%%%%%%%%%%%%%%%%%%%%%%%%%%%%%%%%%%%%%%%%%%%%%
%% Chi
%%%%%%%%%%%%%%%%%%%%%%%%%%%%%%%%%%%%%%%%%%%%%%%%%%%%%%%%%%%%%%%%%

A simplified  formula for $t(p^rq^s,l) $ is given  as follows.

\begin{theorem} \label{chivalue}
    Let $p$ and $q$ be distinct odd primes and let $r,s,$ and $l$ be positive integers.  Then
    \begin{align*} t(p^rq^s,l) = \chi_l(p)t(p^r,l)+\chi_l(q)t(q^s,l)+\chi_l(pq)\sum_{i=1}^{r}\sum_{j=1}^{s} \frac{\phi(p^iq^j)}{ \operatorname{lcm}\left( \operatorname{ord}_{p^i}{(2^l)}, \operatorname{ord}_{q^j}{(2^l)}\right)}.
    \end{align*}
    
\end{theorem}
\begin{proof}
    We note that $\operatorname{ord}_{p^iq^j}(2^l) =  \operatorname{lcm}\left( \operatorname{ord}_{p^i}{(2^l)}, \operatorname{ord}_{q^j}{(2^l)}\right)$. Using \eqref{eq:t} and  Lemma \ref{chipq}, the result follows.
\end{proof}

The next corollary follows from Corollary \ref{lem:odd-gamma}, Lemma \ref{chipq} and Theorem \ref{chivalue}.

%%%%%%%%%%%%%%%%%%%%%%%%%%%%%%%%%%%%%%%%%%%%%%%%%%%%%%%%%%%%%%%%
%% Value of t
%%%%%%%%%%%%%%%%%%%%%%%%%%%%%%%%%%%%%%%%%%%%%%%%%%%%%%%%%%%%%%%%%

\begin{corollary} \label{cor:Epq} Let $p$ and $q$ be distinct odd primes and let $r,s,$ and $l$ be positive integers. Let $\gamma$ and $\beta$ be the integers such that  $2^\gamma|| \operatorname{ord}_p{(2)}$ and  $2^\beta|| \operatorname{ord}_q{(2)}$, respectively.
    Then one of the following statements holds.
    \begin{enumerate}
        \item   If   $2^\gamma | l   $ and   $2^\beta|l  $, then       \begin{align*} t(p^rq^s,l) = t(p^r,l)+t(q^s,l)+\sum_{i=1}^{r}\sum_{j=1}^{s} \frac{\phi(p^iq^j)}{ \operatorname{lcm}\left( \operatorname{ord}_{p^i}{(2^l)}, \operatorname{ord}_{q^j}{(2^l)}\right)}.
        \end{align*}
        \item   If  $2^\gamma | l   $ and   $2^\beta\nmid l  $, then
        \begin{align*} t(p^rq^s,l) = t(p^r,l)+\sum_{i=1}^{r}\sum_{j=1}^{s} \frac{\phi(p^iq^j)}{ \operatorname{lcm}\left( \operatorname{ord}_{p^i}{(2^l)}, \operatorname{ord}_{q^j}{(2^l)}\right)}.
        \end{align*}
        \item   If  $2^\gamma \nmid l   $ and   $2^\beta|l  $,    then
        \begin{align*} t(p^rq^s,l) =  t(q^s,l)+\ \sum_{i=1}^{r}\sum_{j=1}^{s} \frac{\phi(p^iq^j)}{ \operatorname{lcm}\left( \operatorname{ord}_{p^i}{(2^l)}, \operatorname{ord}_{q^j}{(2^l)}\right)}.
        \end{align*}
        \item   If $2^\gamma \nmid l   $ and   $2^\beta\nmid l  $  and $\gamma\ne \beta$, then         \begin{align*} t(p^rq^s,l) =  \sum_{i=1}^{r}\sum_{j=1}^{s} \frac{\phi(p^iq^j)}{ \operatorname{lcm}\left( \operatorname{ord}_{p^i}{(2^l)}, \operatorname{ord}_{q^j}{(2^l)}\right)}.
        \end{align*}
        \item   If  $2^\gamma \nmid l   $ and   $2^\beta\nmid l  $ but $\gamma= \beta$,  then         \begin{align*} t(p^rq^s,l) = 0.
        \end{align*}
    \end{enumerate}
\end{corollary}

It is not difficult to see that the complexity in Corollary \ref{cor:Epq} is lower than a direct computation in \eqref{eq:t}.

%%        \begin{remark} From Corollary \ref{lem:alpha}, we can determine   be the largest  integers  $\alpha$ and $\mu$  such that $p\nmid \operatorname{ord}_{p^\alpha}(2^l)$  and   $p\nmid \operatorname{ord}_{q^\mu}(2^l)$. Using    Lemma \ref{maxalpha},   we can deduce that
%%            \begin{align*}
%%            \sum_{i=1}^{r}\sum_{j=1}^{s} \frac{\phi(p^iq^j)}{ \operatorname{lcm}\left( \operatorname{ord}_{p^i}{(2^l)}, \operatorname{ord}_{q^j}{(2^l)}\right)}=
%%            & \sum_{i=1}^{\alpha}\sum_{j=1}^{\mu} \frac{\phi(p^iq^j)}{ \operatorname{lcm}\left( \operatorname{ord}_{p}{(2^l)}, \operatorname{ord}_{q}{(2^l)}\right)}\\
%%            &+ \sum_{i=1}^{\alpha}\sum_{j=\mu+1}^{s} \frac{\phi(p^iq^j)}{ q^{s-\mu}\operatorname{lcm}\left( \operatorname{ord}_{p}{(2^l)}, \operatorname{ord}_{q}{(2^l)}\right)}\\
%%            &+ \sum_{i=\alpha+1}^{r}\sum_{j=1}^{\mu} \frac{\phi(p^iq^j)}{ p^{r-\alpha}\operatorname{lcm}\left( \operatorname{ord}_{p}{(2^l)}, \operatorname{ord}_{q}{(2^l)}\right)}\\
%%             &+\sum_{i=\alpha+1}^{r}\sum_{j=\mu+1}^{s} \frac{\phi(p^iq^j)}{ p^{r-\alpha}q^{s-\mu}\operatorname{lcm}\left( \operatorname{ord}_{pi}{(2^l)}, \operatorname{ord}_{q}{(2^l)}\right)}.
%%            \end{align*}
%%        \end{remark}

\section{Hermitian Self-Dual Cyclic Codes}
In this section, we focus on the enumeration of Hermitian self-dual cyclic codes of lengths $2^\nu p^r$ and $2^\nu p^rq^s$ over $\mathbb{F}_{2^{2l}}$, where $p$ and $q$ are distinct odd primes and $\nu$, $r$, and $s$ are positive integers.  A simplification of the   formula  for
$\tau (n',l)  $
is established  for all  $n^\prime\in \{p^r, p^rq^s\}$.

\subsection{Number Theoretical Results}
Properties of $\lambda_l$  and $\operatorname{ord}_{p^i}(2^{2l})$ used in the enumeration of Hermitian self-dual cyclic codes are discussed.
\begin{lemma}[{\cite[Theorem 4.1]{JLS2013}}]
    \label{odd-even} Let $j>1$ be an odd integer and let $l$ be a positive integer. Then     $\lambda_l(j)=0$  if and only if $2|| {\rm ord}_{p}(2^l)$ for every prime $p$ dividing~$j$.
\end{lemma}
The next corollary follows immediately from Lemma \ref{odd-even}.
\begin{corollary}\label{char:oddgood} Let $p$  be an odd prime and let $l$ be a positive integer. Then the following statements hold.
    \begin{enumerate}
        \item $\lambda_l(p)=1$ if and only if $\operatorname{ord}_{p}(2^l)$ is odd or $4|\operatorname{ord}_{p}(2^l)$.
        \item $\lambda_l(p)=0$ if and only if   $2||\operatorname{ord}_{p}(2^l)$.
        \item $\lambda_l(p^i)=\lambda_l(p)$   for all positive integers $i$.
        % \item $\chi_l(pq)=1$ if and only if
    \end{enumerate}
\end{corollary}

Next, we determine $\operatorname{ord}_{p^i}(2^{2l})$.
\begin{lemma}\label{order2l}
    Let $p$ be an odd prime  and let $l$ and $i$ be positive integers.  If  $\alpha$  is  the largest  integer such that $p\nmid  \operatorname{ord}_{p^\alpha}(2^l)$, then  one of the following statements holds.
    
    \begin{enumerate}
        \item           If  $\operatorname{ord}_{p}(2^{l})$ is odd, then
        \[ \operatorname{ord}_{p^i}(2^{2l})= \operatorname{ord}_{p^i}(2^{l})=
        \begin{cases}
        {\operatorname{ord}_{p}(2^l)} &\text{ if } i\leq \alpha,   \\
        {p^{i-\alpha}\operatorname{ord}_{p}(2^l)}   &\text{ if } \alpha< i.
        \end{cases}
        \]
        \item              If  $\operatorname{ord}_{p}(2^{l})$ is even, then
        \[ \operatorname{ord}_{p^i}(2^{2l})= \frac{\operatorname{ord}_{p^i}(2^{l})}{2}=
        \begin{cases}
        \frac{\operatorname{ord}_{p}(2^l)}{2} &\text{ if } i\leq \alpha,   \\
        \frac{p^{i-\alpha}\operatorname{ord}_{p}(2^l)}{2}    &\text{ if } \alpha< i.
        \end{cases}
        \]
    \end{enumerate}
    In particular, if $2$ is a primitive root modulo $p^2$, then
    \[ \operatorname{ord}_{p^i}(2^{2l})=    \frac{p^{i-1}(p-1) }{\gcd(p^{i-1}(p-1),2l)}
    \]
    for all positive integers $i$.
\end{lemma}
\begin{proof}   From Lemma \ref{maxalpha},  $\operatorname{ord}_{p}(2^{l})$  and  $\operatorname{ord}_{p^i}(2^{l})$  have the same parity for all positive integers $i$. Since $\operatorname{ord}_{p^i}(2^{2l})=\frac{\operatorname{ord}_{p^i}(2^{l})}{\gcd( \operatorname{ord}_{p^i}(2^{l}),2)}$  and
    \[\gcd( \operatorname{ord}_{p^i}(2^{l}),2)=\begin{cases}
    1 &\text{ if }  \operatorname{ord}_{p}(2^{l}) \text{ is odd},\\
    2 &\text{ if }  \operatorname{ord}_{p}(2^{l}) \text{ is even},
    \end{cases}\]
    the results follow from Lemma \ref{char:good}.
\end{proof}

\subsection{Hermitian Self-Dual Cyclic Codes of Length $2^\nu p^r$}
In this subsection,  an explicit formula for the number of Hermitian self-dual cyclic codes of length $2^\nu p^r $  over $\mathbb{F}_{2^{2l}}$ is given  together with  an efficient  algorithm to compute the number of such self-dual codes.

\begin{theorem}\label{tau} Let $p$ be an odd prime  and $r$ be a positive integer. Let $\alpha $ be the largest positive integer such that $p\nmid \operatorname{ord}_{p^\alpha}(2^l)$ and let  $\gamma$ be the  integer such that $2^\gamma || \operatorname{ord}_{p}(2)$ .   Then
    %        \begin{align}
    %        \label{eq:simTau-old} \tau(p^r,l) =\begin{cases}
    %        \frac{ \gcd(\operatorname{ord}_{p}(2),l)}{2{\operatorname{ord}_{p}(2)}} \left(  {p^{\alpha}-1}+  \frac{p^\alpha(p^{r-\alpha}-1)}{p^{r-\alpha }}  \right) & \text{ if }    2^\gamma | l,\\
    %        \frac{ \gcd(\operatorname{ord}_{p}(2),l)}{{\operatorname{ord}_{p}(2)}} \left(  {p^{\alpha}-1}+  \frac{p^\alpha(p^{r-\alpha}-1)}{p^{r-\alpha }}  \right) & \text{ if }      2^{\gamma-1} \nmid l,  \\
    %        0& \text{ if }     2^{\gamma-1} ||  l.
    %        \end{cases}\end{align}
    {      \begin{align}
        \label{eq:simTau} \tau(p^r,l) =\begin{cases}
        \frac{ \gcd(\operatorname{ord}_{p}(2),l)}{2{\operatorname{ord}_{p}(2)}} \left(  {p^{\alpha}-1}+   (p-1)(r-\alpha)p^{\alpha-1} \right)& \text{ if }    2^\gamma | l,\\
        \frac{ \gcd(\operatorname{ord}_{p}(2),l)}{{\operatorname{ord}_{p}(2)}} \left(  {p^{\alpha}-1}+   (p-1)(r-\alpha)p^{\alpha-1} \right)& \text{ if }      2^{\gamma-1} \nmid l,  \\
        0& \text{ if }     2^{\gamma-1} ||  l.
        \end{cases}\end{align}
    }
    {  In particular, if $2$ is a primitive root modulo $p^2$, then for any positive integer $r$
        \begin{align} \label{simTauPri} \tau(p^r,l) =\begin{cases}
        \frac{1}{2}\left( \sum_{i=1}^{r} \operatorname{gcd}(p^{i-1}(p-1),2l)\right)   & \text{ if }   2^{\gamma} |  l  \text{ or }2^{\gamma-1} \nmid  l ,\\
        0& \text{ if }   2^{\gamma-1} || l.
        \end{cases}\end{align}}
    
\end{theorem}

\begin{proof} From  Corollary \ref{lem:odd-gamma},  $ 2^{\gamma-1} || l $ if and only if  $2|| \operatorname{ord}_{p}(2^l)$ which is equivalent to    $\chi_l(p)=0$ by Corollary  \ref{char:oddgood}. In this case,   $\lambda_l(p^i)=0$ for all $1\leq i\leq r$. Equivalently, $\tau(p^r,l) =0$ if and only if $ 2^{\gamma-1} || l $.
    
    Assume that   $2^\gamma | l$ or $ 2^{\gamma-1} \nmid  l$. By Corollary \ref{lem:odd-gamma}, it follows  that  $\operatorname{ord}_{p}(2^l)$ is odd or $4|\operatorname{ord}_{p}(2^l)$.  %Then  $\lambda_l(p^i)=1$ for all $1\leq i\leq r$ by Corollary  \ref{char:oddgood}.
    Consider the following two cases.
    
    \noindent{\bf Case 1:}   $\operatorname{ord}_{p}(2^l)$ is odd.  From the definition of $\tau$ in  \eqref{eq:tau} and     Lemma \ref{order2l}, we have
    \begin{align} \tau (p^r,l) &= \frac{1}{2} \sum_{i=0}^{r} \lambda_l(p^i)\frac{\phi(p^i)}{\operatorname{ord}_{p^i}(2^{2l})}\\
    % &=\frac{1}{2} \sum_{i=1}^{r} \frac{p^{i-1}(p-1)}{\operatorname{ord}_{p^i}(2^{2l})}\\
    &=\frac{1}{2} \sum_{i=1}^{r} \frac{p^{i-1}(p-1)}{\operatorname{ord}_{p^i}(2^{l})}\\
    %    &=\frac{ \gcd(\operatorname{ord}_{p}(2),l)}{2{\operatorname{ord}_{p}(2)}} \left(  {p^{\alpha}-1}+  \frac{p^\alpha(p^{r-\alpha}-1)}{p^{r-\alpha }}  \right)\\
    &{=\frac{ \gcd(\operatorname{ord}_{p}(2),l)}{2{\operatorname{ord}_{p}(2)}} \left(  {p^{\alpha}-1}+   (p-1)(r-\alpha)p^{\alpha-1} \right).}
    \end{align}

    \noindent{\bf Case 2:}   $4|\operatorname{ord}_{p}(2^l)$.
    From  \eqref{eq:tau}    and  Lemma \ref{order2l},  it follows that
    \begin{align} \tau (p^r,l) &= \frac{1}{2} \sum_{i=0}^{r} \lambda_l(p^i)\frac{\phi(p^i)}{\operatorname{ord}_{p^i}(2^{2l})}\\
    % &=\frac{1}{2} \sum_{i=1}^{r} \frac{p^{i-1}(p-1)}{\operatorname{ord}_{p^i}(2^{2l})}\\
    &=\frac{1}{2} \sum_{i=1}^{r} \frac{2p^{i-1}(p-1)}{\operatorname{ord}_{p^i}(2^{l})}\\
    %    &=\frac{ \gcd(\operatorname{ord}_{p}(2),l)}{{\operatorname{ord}_{p}(2)}} \left(  {p^{\alpha}-1}+  \frac{p^\alpha(p^{r-\alpha}-1)}{p^{r-\alpha }}  \right)\\
    &{=\frac{ \gcd(\operatorname{ord}_{p}(2),l)}{{\operatorname{ord}_{p}(2)}} \left(  {p^{\alpha}-1}+   (p-1)(r-\alpha)p^{\alpha-1} \right).}
    \end{align}
    
    Finally, assume further  that $2$ is a primitive root modulo $p^2$. By Lemma \ref{maxalpha}, we have  $\operatorname{ord}_{p^i}(2^l)=    \frac{p^{i-1}(p-1) }{\gcd(p^{i-1}(p-1),l)}$ for all $1\leq  i \leq r$.   Hence,
    \begin{align} \tau(p^r,l) &= \frac{1}{2} \sum_{i=0}^{r} \chi_l(p^i)\frac{\phi(p^i)}{\operatorname{ord}_{p^i}(2^{2l})}\\
    &=\frac{1}{2} \sum_{i=1}^{r} \frac{p^{i-1}(p-1)}{ \frac{p^{i-1}(p-1) }{\gcd(p^{i-1}(p-1),2l)}}\\
    &= \frac{1}{2} \sum_{i=1}^{r} \gcd(p^{i-1}(p-1),2l)
    \end{align}
    as desired.
\end{proof}

\begin{corollary}
    Let $l$ be a positive integer and $p$ be an odd prime satisfying $l$ is not divisible by $p$. If $2$ is a primitive root modulo $p^2$ and
    $\operatorname{ord}_p{(2^l)}$ is odd or $\operatorname{ord}_p{(2^l)}$ is divisible by $4$,  then
    $\tau(p^r,l) = r\operatorname{gcd}(\frac{p-1}{2},l)$ for all positive integers $r$.
\end{corollary}
\begin{example}\label{example_tau} Let $p=11$, $l=4$ and $r=3$. We can see that $2$ is a primitive root modulo $11^2$ and $\operatorname{ord}_{11}{(2^4)} = 5$ is odd.
    Thus $\tau(11^3,4) = 3\operatorname{gcd}(5,4)$=3.
\end{example}

The results on    $\tau(p^r,l)$  discussed above can be  summarized in  Algorithm \ref{al3}.

\begin{figure}[!hbt]
    \centering
    \parbox{12cm}{ \hrule \vskip1em

        For  an odd prime $p$ and positive integers $l$ and $r$, do the following steps.
        \begin{enumerate}
            \item     Compute $\operatorname{ord}_{p}{(2)}$.
            \item   Determine $\gamma$ such that  $2^\gamma || \operatorname{ord}_{p}{(2)}$.
            \begin{enumerate}
                \item[2.1]  If $2^{\gamma-1} || l$, then $\tau(p^r,l)=0$ by \eqref{eq:simTau}. Done.
                \item[2.2] If $2^\gamma| l$ or $2^{\gamma -1}\nmid l$, then compute $\operatorname{ord}_{p^2}{(2)}$.
                \begin{enumerate}
                    \item[2.2.1]  If  $\operatorname{ord}_{p^2}{(2)}=p(p-1)$, then evaluate \eqref{simTauPri}. Done.
                    \item[2.2.2]   If  $\operatorname{ord}_{p^2}{(2)}\ne p(p-1)$, then do the following steps.
                    \begin{enumerate}[i)]
                        \item Compute $\operatorname{ord}_{p^r}{(2^l)}$.
                        \item Determine  the largest
                        integer $\alpha$ such that $p \nmid  \operatorname{ord}_{p^\alpha}(2^l)$  (by Corollary \ref{lem:alpha}).
                        \item Evaluate   \eqref{eq:simTau}. Done.
                    \end{enumerate}
                \end{enumerate}
            \end{enumerate}
        \end{enumerate}
        \hrule
    }
    \caption{Steps in  Computing $\tau(p^r,l)$}     \label{al3}
\end{figure}

Similar to the Euclidean case, a direct computation of  $\tau(p^r,l)$  from \eqref{eq:t} requires the values of  $\chi_l(p^i), \phi(p^i)$ and $\operatorname{ord}_{p^i}{(2^l)}$ for all $0\leq i\leq r$.  The complexity can be reduced  using   Algorithm \ref{al3}   since it  requires  only $\operatorname{ord}_{p}{(2)}$, $\operatorname{ord}_{p^2}{(2)}$, $\operatorname{ord}_{p^r}{(2^l)}$ and some basic expressions in \eqref{simTauPri} or \eqref{eq:simTau}.

From Theorems \ref{ESD-form} and \ref{HSD-form},   a Euclidean self-dual cyclic code over  $\mathbb{F}_{2^l}$  exists  for all integers $l$ but an Hermitian  self-dual cyclic code over  $\mathbb{F}_{2^l}$  exists if and only if $l$ is even.   Over $\mathbb{F}_{2^{2l}}$,   both the  Euclidean and Hermitian self-dual cyclic codes always exist.  The numbers of self-dual codes in the two families can be  compared in terms of $t(p^r,2l) $ and $\tau(p^r,l)$ as follows.
\begin{proposition}
    Let $p$ be an odd prime  and $r$ be positive integers. Let   $\gamma$ be the  integer such that $2^\gamma || \operatorname{ord}_{p}(2)$ .   Then one of the following holds
    \begin{enumerate}
        \item   If  $2^\gamma | l$, then  $\tau(p^r,l) = t(p^r,2l) $.
        \item  If $ 2^{\gamma-1} ||  l$, then $0  = \tau(p^r,l) < t(p^r,2l)$.
        \item   If $ 2^{\gamma-1} \nmid l $, then  $\tau(p^r,l) >  t(p^r,2l) = 0$ .
    \end{enumerate}
\end{proposition}
\begin{proof}
    For the first part, it suffices to show that $\operatorname{gcd}(\operatorname{ord}_p(2),l) = \operatorname{gcd}(\operatorname{ord}_p(2),2l)$. Since $2^\gamma || \operatorname{ord}_{p}(2)$ and $2^\gamma |l$, we have
    $\operatorname{gcd}(\operatorname{ord}_p(2),l) = \operatorname{gcd}(\operatorname{ord}_p(2),2l)$ as desired.
    
    For the rest of the theorem, it follows easily from Theorems \ref{tvalue} and \ref{tau}.
\end{proof}

\subsection{Hermitian  Self-Dual Cyclic Codes of Length $2^\nu p^rq^s$}

From Lemma \ref{odd-even}, we have the following lemma.
\begin{lemma}\label{lem3.6}
    Let $p$ and $q$  be distinct odd primes and let $l$ be a positive integer. Then the following statements hold.
    \begin{enumerate}
        \item $\lambda_l(pq)=1$ if and only if  $\lambda_l(p)=1$ or $\lambda_l(q)=1$.
        \item $\lambda_l(p^iq^j)=\lambda(pq)$ for all positive integers $i$ and $j$.
    \end{enumerate}
\end{lemma}

\begin{theorem} \label{thm3.7} Let $p$ and $q$ be distinct odd primes and let $r,s,$ and $l$ be positive integers.  Then
    \begin{align*} \tau(p^rq^s,l) = \lambda_l(p)\tau(p^r,l)+\lambda_l(q)\tau(q^s,l)+\lambda_l(pq)\sum_{i=1}^{r}\sum_{j=1}^{s} \frac{\phi(p^iq^j)}{ \operatorname{lcm}\left( \operatorname{ord}_{p^i}{(2^{2l})}, \operatorname{ord}_{q^j}{(2^{2l})}\right)}.
    \end{align*}
    
\end{theorem}
\begin{proof}
    We note that $\operatorname{ord}_{p^iq^j}(2^{2l}) =  \operatorname{lcm}\left( \operatorname{ord}_{p^i}{(2^{2l})}, \operatorname{ord}_{q^j}{(2^{2l})}\right)$. The theorem follows from   \eqref{eq:tau} and Lemma \ref{lem3.6}.
\end{proof}

%%%%%%%%%%%%%%%%%%%%%%%%%%%%%%%%%%%%%%%%%%%%%%%%%%%%%%%%%%%%%%%%%%%%%%%%%%%%%%%%%%%%%%

%%%%%%%%%%%%%%%%%%%%%%%%%%%%%%%%%%%%%%%%%%%%%%%%%%%%%%%%%%%%%%%%%%%%%%%%%%%%%%%%%%%%%%
The next corollary follows from Corollary \ref{lem:odd-gamma}, Lemma \ref{lem3.6} and Theorem \ref{thm3.7}.

\begin{corollary} Let $p$ and $q$ be distinct odd primes and let $r,s,$ and $l$ be positive integers. Let $\gamma$ and $\beta$ be the integers such that  $2^\gamma|| \operatorname{ord}_p{(2)}$ and  $2^\beta|| \operatorname{ord}_q{(2)}$, respectively.
    Then one of the following statements holds.
    \begin{enumerate}
        \item   If {$2^\gamma|l$} or $2^{\gamma-1}\nmid l$, and  $ 2^\beta |l$ or $2^{\beta-1}\nmid l$, then       \begin{align*} \tau(p^rq^s,l) = \tau(p^r,l)+\tau(q^s,l)+\sum_{i=1}^{r}\sum_{j=1}^{s} \frac{\phi(p^iq^j)}{ \operatorname{lcm}\left( \operatorname{ord}_{p^i}{(2^{2l})}, \operatorname{ord}_{q^j}{(2^{2l})}\right)}.
        \end{align*}
        \item   If  {$2^\gamma|l$} or $2^{\gamma-1}\nmid l$, and  $2^{\beta-1}||l$, then
        \begin{align*} \tau(p^rq^s,l) = \tau(p^r,l)+\sum_{i=1}^{r}\sum_{j=1}^{s} \frac{\phi(p^iq^j)}{ \operatorname{lcm}\left( \operatorname{ord}_{p^i}{(2^{2l})}, \operatorname{ord}_{q^j}{(2^{2l})}\right)}.
        \end{align*}
        \item   If   $2^{\gamma-1}||l$, and {$ 2^\beta|l$} or $2^{\beta-1}\nmid l$,  then
        \begin{align*} \tau(p^rq^s,l) =  \tau(q^s,l)+\ \sum_{i=1}^{r}\sum_{j=1}^{s} \frac{\phi(p^iq^j)}{ \operatorname{lcm}\left( \operatorname{ord}_{p^i}{(2^{2l})}, \operatorname{ord}_{q^j}{(2^{2l})}\right)}.
        \end{align*}
        
        \item   If $2^{\gamma-1}||l$ and $2^{\beta-1}||l$, then \begin{align*} \tau(p^rq^s,l) = 0.
        \end{align*}
    \end{enumerate}
\end{corollary}

Similar to the Euclidean case,    the complexity of  the direct computation of  $\tau(p^r,l)$  in \eqref{eq:tau} can be reduced using Corollary \ref{cor:Epq}.

%%    {\color{blue}
%%        \begin{remark}
%%            From   Lemma \ref{maxalpha},   \[ \sum_{i=1}^{r}\sum_{j=1}^{s} \frac{\phi(p^iq^j)}{ \operatorname{lcm}\left( \operatorname{ord}_{p^i}{(2^{2l})}, \operatorname{ord}_{q^j}{(2^{2l})}\right)}.\] can be simplified LATER.
%%        \end{remark}
%%
%%    }
%%
%%    \begin{figure}[!hbt]
%%        \centering
%%        \parbox{13cm}{ \hrule \vskip1em
%%
%%            {\color{blue} TO BE PROVIDED WITH IT EFFICIENCY ANALYSIS.}
%%
%%            \vskip1em
%%            \hrule
%%
%%        }
%%        \caption{Steps in  Computing $\tau(p^rq^s,l)$}     \label{al4}
%%    \end{figure}
%%

%\section{Conclusion}

%Note that this idea may apply for arbitrary even integers $n$. However,  the formula would look tedious if $n$ contains more than two prime factors.

\section*{Acknowledgments}
{This research was supported by the Thailand Research Fund and the Office of Higher Education Commission  of Thailand under Research
    Grant MRG6080012.}


\begin{thebibliography}{99}
    
    \bibitem{AB1959} A. Brauer,  A note on a number theoretical paper of Sierpinski,  {\em Proc. Amer. Math. Soc.}  \textbf{11}  (1960) 406--409.
    
    
    
    \bibitem{CDL2014}  B. Chen, H. Q. Dinh, H. Liu, Repeated-root constacyclic codes of length $\ell p^s$ and their duals,   {\em Discrete Appl. Math.}
    \textbf{177} (2014)  60--70.
    
    \bibitem{CLZ2016} B. Chen, S. Ling, G. Zhang, Enumeration formulas for self-dual cyclic codes, {\em Finite Fields Appl.} \textbf{42}  (2016) 1--22.
    
    \bibitem{JLX2011} Y. Jia, S. Ling, and C. Xing, On self-dual cyclic codes over finite fields,   {\em IEEE Trans. Inf. Theory}   \textbf{57} (2011) 2243--2251. 
    
    
    \bibitem{JLS2013} S, Jitman, S. Ling,  P. Sol\'e,  Hermitian self-dual abelian codes,   {\em IEEE Trans. Inf. Theory}    \textbf{60} (2014) 1496 --1507.
    
    
    \bibitem{KZ2008} X. Kai, S. Zhu, On cyclic self-dual codes,  {\em Appl. Algebra Eng., Commun. Comput.} \textbf{19} (2008) 509--525.
    
    
    \bibitem{M1997}  P. Moree, On the divisors of $a^k+b^k$,   {\em Acta Arithmetica} \textbf{80} (1997)  197--212.
    % \bibitem{MC1999} P. Moree and J. Cazaran,  ``On a claim of Ramanujan in his first letter to Hardy,''  {\em Exposition. Math.,} vol.  17,   pp. 289--311, 1999.
    
    \bibitem{N2000} M. B. Nathanson, {\em Elementary Methods in Number Theory}, (Springer, 2000).
    
    \bibitem{NRS2006}  G. Nebe, E. M. Rains, and N. J. A. {\em Sloane, Self-Dual Codes and  Invariant Theory} (Algorithms and Computation in Mathematics), vol. 17.
    Berlin Heidelberg, Germany: Springer-Verlag, 2006.
    
    
    \bibitem{WS1958} W. Sierpinski,  Sur une d{\'e}composition des nombres premier en deux classes,  {\em Collect. Math.}  \textbf{10} (1958)  81--83.
    
    
    \bibitem{S2003}   G. Skersys,  The average dimension of the hull of cyclic codes,   {\em Discrete Appl. Math.}  \textbf{128} ( 2003) 275--292.
    
    
    % \bibitem{W2002} W. Willems,  ``A note on self-dual group codes,''  {\em IEEE Trans. Inform. Theory,} vol.  48, pp.  3107--3109,  2002.
    
    
    
\end{thebibliography}
\end{document}